\documentclass[pra,showkeys,showpacs,superscriptaddress,nofootinbib,twocolumn]{revtex4-2}

\usepackage[lmargin=.7in,rmargin=.7in,tmargin=.7in,bmargin=1in]{geometry}
\usepackage[utf8]{inputenc}

\usepackage[caption=false]{subfig}
\usepackage{graphicx}
\usepackage{epstopdf}
\usepackage{array}
\usepackage{verbatim}
\usepackage{amsmath,amsfonts,amssymb,amscd,mathtools}
\usepackage{dsfont}
\usepackage{physics}
\usepackage{tabularx}
\usepackage{stmaryrd}
\usepackage{enumerate}

 \usepackage[colorlinks=true,        
            allcolors = black,  
            citecolor=DarkBlue, 
            linkcolor=DarkBlue,
            urlcolor=DarkBlue
        ]{hyperref}

\usepackage{wasysym}
\usepackage{braket}
\usepackage{mathrsfs}
\usepackage{microtype}
\usepackage{hyperref}
\usepackage[dvipsnames]{xcolor}

\hypersetup{
    bookmarksnumbered=true, 
    unicode=false, 
    pdfstartview={FitH}, 
    pdftitle={}, 
    pdfauthor={}, 
    pdfsubject={}, 
    pdfcreator={}, 
    pdfproducer={}, 
    pdfkeywords={}, 
    pdfnewwindow=true, 
    colorlinks=true, 
    linkcolor=NavyBlue, 
    citecolor=NavyBlue, 
    filecolor=NavyBlue, 
    urlcolor=NavyBlue 
}
\usepackage{diagbox}
\usepackage{tikz}
\usepackage{newtxtext,newtxmath}

\usepackage{algorithm}
\usepackage{algorithmicx}
\usepackage{algpseudocode}
\usepackage{amsthm}



\usepackage{booktabs}

\usepackage{aliascnt}
\theoremstyle{plain}
\newtheorem{thm}{Theorem}
\newaliascnt{cor}{thm}

\aliascntresetthe{cor}
\newaliascnt{lem}{thm}
\newtheorem{lem}[lem]{Lemma}
\aliascntresetthe{lem}
\newaliascnt{pro}{thm}
\newtheorem{pro}[pro]{Proposition}
\aliascntresetthe{pro}
\theoremstyle{definition}
\newaliascnt{defn}{thm}
\newtheorem{defn}[defn]{Definition}
\aliascntresetthe{defn}
\newaliascnt{remark}{thm}

\aliascntresetthe{remark}
\newaliascnt{ex}{thm}

\aliascntresetthe{ex}
\newaliascnt{question}{thm}

\aliascntresetthe{question}
\newaliascnt{problem}{thm}

\aliascntresetthe{problem}

\usepackage{cleveref}
\crefname{thm}{Theorem}{Theorems}
\crefname{cor}{Corollary}{Corollaries}
\crefname{lem}{Lemma}{Lemmas}
\crefname{pro}{Proposition}{Propositions}
\crefname{defn}{Definition}{Definitions}
\crefname{remark}{Remark}{Remarks}
\crefname{ex}{Example}{Examples}
\crefname{question}{Question}{Questions}
\crefname{problem}{Problem}{Problems}
\Crefname{equation}{Equation}{Equations}
\crefname{figure}{Fig.}{Figs.}
\Crefname{figure}{Figure}{Figures}
\crefname{equation}{Eq.}{Eqs.}
\Crefname{equation}{Equation}{Equations}
\crefname{figure}{Fig.}{Figs.}
\Crefname{figure}{Figure}{Figures}
\crefname{section}{Sec.}{Secs.}
\Crefname{section}{Section}{Sections}

\newcommand{\eq}[1]{(\hyperref[eq:#1]{\ref*{eq:#1}})}

\renewcommand{\sec}[1]{\hyperref[sec:#1]{Section~\ref*{sec:#1}}}
\newcommand{\thrm}[1]{\hyperref[thrm:#1]{Theorem~\ref*{thrm:#1}}}
\newcommand{\lemm}[1]{\hyperref[lemm:#1]{Lemma~\ref*{lemm:#1}}}
\newcommand{\prop}[1]{\hyperref[prop:#1]{Proposition~\ref*{prop:#1}}}
\newcommand{\corr}[1]{\hyperref[corr:#1]{Corollary~\ref*{corr:#1}}}
\newcommand{\fig}[1]{\hyperref[fig:#1]{~\ref*{fig:#1}}}
\newcommand{\deff}[1]{\hyperref[deff:#1]{~\ref*{deff:#1}}}

\newcommand{\mE}{\mathcal{E}}

\newcommand{\mF}{\mathcal{F}}

\newcommand{\mO}{\mathcal{O}}

\newcommand{\mbE}{\mathbb{E}}

\newcommand{\mbR}{\mathbb{R}}
\newcommand{\mbV}{\mathbb{V}}

\DeclareMathAlphabet{\matheu}{U}{eus}{m}{n}

\newcommand{\ba}{\begin{eqnarray}}
\newcommand{\ea}{\end{eqnarray}}
\newcommand{\bann}{\begin{eqnarray*}}
\newcommand{\eann}{\end{eqnarray*}}
\newcommand{\bal}{\begin{equation}\begin{aligned}}
\newcommand{\eal}{\end{aligned}\end{equation}}

\newcolumntype{L}[1]{>{\raggedright}p{#1}}
\newcolumntype{C}[1]{>{\centering}p{#1}}
\newcolumntype{R}[1]{>{\raggedleft}p{#1}}
\newcolumntype{D}{>{\centering\arraybackslash}X}

\DeclareMathOperator{\Prob}{Prob}

\begin{document}

\title{Weak distillation of quantum resources}

\author{Shinnosuke Onishi}
\email{shinonishi@g.ecc.u-tokyo.ac.jp}
\affiliation{Department of Integrated Sciences, The University of Tokyo, Tokyo 153-8902, Japan}

\author{Oliver Hahn}
\email{hahn@g.ecc.u-tokyo.ac.jp}
\affiliation{Department of Basic Science, The University of Tokyo, Tokyo 153-8902, Japan}

\author{Ryuji Takagi}
\email{ryujitakagi@g.ecc.u-tokyo.ac.jp}
\affiliation{Department of Basic Science, The University of Tokyo, Tokyo 153-8902, Japan}

\begin{abstract}
Importance sampling based on quasi-probability decomposition is the backbone of many widely used techniques, such as error mitigation, circuit knitting, and, more generally, virtual quantum resource distillation, as it allows one to simulate operations that are not accessible in a given setting. However, this class of protocols faces a fundamental problem---it only allows to estimate expectation values. Here, we provide a general framework that lifts any quasi-probability-based protocol from expectation value estimation to a weak simulator, realizing sampling from the desired distribution only using a restricted class of quantum resources. Our method runs with the sampling cost proportional to the negativity of the quasi-probability, in stark contrast to the naive estimation-based approach that requires a large number of samples even in the case of small negativity. We show that our method requires significantly fewer samples in a number of relevant scenarios, such as error mitigation, entanglement distillation and magic state distillation. Our framework realizes the weak simulation of quantum resources without actually distilling the state, introducing a new notion of quantum resource distillation.
\end{abstract}

\maketitle

While recent technological advances offer growing potential of quantum information processing, the actual realization of quantum advantages faces inevitable challenges due to operational restrictions and experimental limitations. It is therefore of paramount importance to obtain higher-quality quantum resources~\cite{Chitambar2019quantum}, such as entanglement~\cite{RevModPhys.81.865, PhysRevLett.78.2275}, magic~\cite{Veitch_2014,Howard2017application} and coherence~\cite{RevModPhys.89.041003} among many more, from available components in order to perform a given task.

Resource distillation~\cite{PhysRevA.53.2046,PhysRevLett.81.2839, HORODECKI_2012} is one of the most common methods to obtain a resource state using low-quality states available under the given operational capabilities. 
However, distillation protocols typically require many state copies and accurate coherent interactions between them, which are practically challenging. 
To address this problem, the framework of virtual quantum resource distillation has recently been proposed~\cite{Yuan2024virtual, Takagi2024virtual}.
The underlying idea is to represent a resourceful quantum state not available in the given setting by a quasi-probability distribution of the available low-quality states.
This allows one to employ the importance sampling technique to estimate an arbitrary expectation value of the desired resourceful state by only using low-quality states and classical post-processing. 
The additional sampling overhead scales with the amount of negativity needed in the quasi-probability decomposition, which can be understood as a quantifier of how resourceful the target state is with respect to the set of easily prepared states in the setting~\cite{PhysRevA.59.141, Howard2017application, Regula2018convex}. This circumvents the need to directly prepare the high-quality quantum state, making the protocol practically appealing.
Furthermore, the framework of virtual quantum resource distillation is highly general and encompasses various techniques---such as error mitigation~\cite{PhysRevLett.119.180509,Cai2023quantum}, circuit knitting~\cite{Mitarai_2021,Piveteau2024circuit}, implementation of unphysical maps~\cite{Jiang2021physical,Regula2021operational}, and classical simulation of quantum circuits~\cite{Mari2012positive,Pashayan2015estimating,Howard2017application,PRXQuantum.6.010330}---in the unified framework.

Despite the versatility of this approach, quasi-probability-based methods come with a severe restriction---they only admit expectation value estimation. 
Although this is sufficient for several applications~\cite{Cerezo2021VariationalQuantumAlgorithms}, many important algorithms, such as Shor's algorithm for factoring~\cite{shor} and sampling tasks for certification of quantum advantages~\cite{10.1098/rspa.2008.0443,10.1098/rspa.2010.0301,Boixo2018CharacterizingQuantumSupremacy,Arute2019QuantumSupremacy}, require direct access to samples of the measurement outcomes. 
Indeed, the expectation value estimation is known to be a weaker form of simulation compared to the capability of sampling from the desired distribution---known as \emph{weak simulation}~\cite{Pashayan2020fromestimationof}. 
Despite its confusing terminology, weak simulation is the natural and strong notion of simulation, in the sense that it directly mimics what quantum computers accomplish.
Therefore, asking whether virtual quantum resource distillation, a quasi-probability-based method to simulate resourceful objects, can be extended to support sampling tasks is well motivated.

One direct approach would be to estimate the full probability distribution by the quasi-probability-based method and then sample from it, as recently discussed in the context of error mitigation~\cite{Liu2025quantumerrormitigationsampling}. 
However, this strategy fails to naturally reflect the excess resource contents of the target state with respect to the given setting.
In fact, even when the desired state is close to an available one with small negativity in the quasi-probability decomposition, it incurs a substantial overhead for estimating the entire probability distribution.
Considering that there is no extra cost to pay when there is no negativity, in which case one can directly take a sample from the given device, this naive approach shows an unnatural discontinuity in the simulation cost, particularly in the small negativity regime. 
Therefore, an efficient way that naturally extends the quasi-probability-based framework to allow for sampling access is currently missing.

In this work, we fill this gap by introducing a general framework of \emph{weak resource distillation}, which lifts an expectation value estimation based on quasi-probability to weak simulation. 
Analogous to virtual quantum resource distillation, our framework enables sampling without explicitly preparing the high-quality state, thereby removing the requirement for large quantum memory.
The sampling cost of our methods scales directly with the negativity of the quasi-probability, not only circumventing the aforementioned discontinuity in the small negativity regime, but also providing a novel operational interpretation of the negativity of the quasi-probability decomposition, also known as the robustness measure of quantum resources~\cite{PhysRevA.59.141,Howard2017application,Regula2018convex}. We further support the practical significance of our approach by numerical simulation and find that our method requires significantly fewer samples compared to the naive probability-estimation-based approach in several important scenarios such as error mitigation, entanglement distillation, and magic state distillation. 
Together with the high generality of our framework, our results open the door to use the power of quasi-probability methods in many relevant scenarios.

\vspace{0.2cm}
\textbf{\emph{Preliminaries}.\,---} 
We consider the setting where one aims to run a protocol that requires a certain state under physical restriction.
The framework of resource theories~\cite{Chitambar2019quantum} lends itself to such scenarios.
In particular, we let $\mO$ denote the set of operations that are accessible in the given setting and $\mF$ denote the set of quantum states that can be prepared by the operation in $\mO$. Our framework is fully general as long as the sets $\mO$ and $\mF$ are finite dimensional and convex.
We also note that, by definition of $\mF$, it is ensured that an operation in $\mO$ can never map a state in $\mF$ to another state outside of $\mF$, reflecting the natural operational restriction.

If the desired state $\rho$ required in the protocol is not directly available in the given setting, i.e., $\rho\not\in\mF$, one can aim to simulate the process by only using states in $\mF$.
The framework of the virtual quantum resource distillation~\cite{Yuan2024virtual,Takagi2024virtual} provides a way of accomplishing this when the goal of the tasks is to estimate the expectation value of some observable.
The core idea is to express the resource state $\rho$ in terms of the linear combination of the states in $\mF$ as $\rho=\sum_i c_i \sigma_i$, $\sigma_i\in\mF$, $c_i\in\mbR$.  
The crucial part is that some $c_i$'s can be negative, allowing it to express any state in the linear span of $\mF$, which is typically much larger than the set $\mF$ itself. 
Noting also the unit trace condition for quantum states, $c_i$'s need to be normalized as $\sum_i c_i = 1$.
This ensures that $\{c_i\}_i$ constructs a \emph{quasi-probability} distribution, which sums to one but each value can be negative. 
Because of the convexity of $\mF$, the quasi-probability decomposition can generally be written in the form $\rho = c_+ \sigma_+ - c_- \sigma_-$
with positive constants $c_\pm\geq 0$ satisfying $c_+-c_-=1$ and states $\sigma_\pm\in \mF$. 
Given the quasi-probability decomposition, the expectation value estimation can be done by employing importance sampling, using the number of samples that scales with the absolute sum $\gamma\coloneqq c_+ + c_- = 1 + 2c_-$ of the quasi-probability, which grows with the size of the negativity $c_-$ in the quasi-probability decomposition. 
As the quasi-probability decomposition is generally not unique, it is meaningful to consider the optimal decomposition that minimizes the negativity.
This minimized negativity is also known as the robustness measure of quantum resources~\cite{PhysRevA.59.141,Howard2017application, Regula2018convex}, which serves as a resource quantifier with respect to the set $\mF$ and $\mO$ of free states and operations. 
We refer readers to \cref{sup: virtual distillation framework} of the Supplemental Material for more detailed descriptions of the virtual quantum resource distillation.

This framework encompasses many scenarios by flexibly choosing the set $\mF$, e.g., stabilizer or positive Wigner states for classical simulation of quantum circuits~\cite{Howard2017application,Bravyi2016improved, Seddon2021quantifying}, separable states for entanglement distillation~\cite{RevModPhys.81.865}, the states preparable by local subcircuits for circuit knitting~\cite{Mitarai_2021,10236453}, and the states produced by the inverse map of the noise channel for an error mitigation method known as probabilistic error cancellation~\cite{PhysRevLett.119.180509,Endo2018practical,Takagi2021optimal}. 
Importantly, this framework goes beyond these standard examples when some degrees of resources can be created in the given setting. For instance, in the context of entanglement distillation, if we have access to a device that can create weakly entangled states, the virtual quantum resource distillation admits more efficient entanglement distillation by taking the set of weakly entangled states for $\mF$ rather than the set of separable states, allowing us to incorporate the full capability of the given setting.

\vspace{0.2cm}
\textbf{\emph{Weak resource distillation}.\,---} 
We now aim to lift the expectation value estimation realized in the virtual quantum resource distillation to the weak simulation of the process involving the resource state. 
Recall that weak simulation aims to obtain a sample from a distribution $\{\tilde p_x\}_x$ that is close to the true distribution $\{p_x\}_x$ in the total variation distance as $d_{\rm tv}(p,\tilde p)=\frac{1}{2}\sum_x |p_x - \tilde p_x|\leq \epsilon$ with probability at least $1-\delta$, where $\epsilon$ and $\delta$ are some small constants characterizing the performance of the simulation. 
In the context of resource distillation, we are generally interested in the samples from a resource state $\rho$.
The true probability distribution in this case is $p_x = \braket{x|\rho|x}$ where $\{\ket{x}\}_x$ is the computational basis.
We aim to get a sample from the distribution close to $p_x$ by using $N$ samples of the states in $\mF$. 
If $\rho\in\mF$, the sampling cost $N$ is trivially one, because the state $\rho\in\mF$ can directly be prepared by the given device. 
On the other hand, if $\rho\not\in\mF$, one generally needs a larger number $N$ of samples from the available states $\mF$ to construct a desired sampler.
We call this sampling process \emph{weak distillation} of a resource state $\rho$ and call the required number $N$ the \emph{sampling cost} for weak resource distillation. 
The major question then is to characterize the sampling cost $N$ with respect to the target trace distance error $\epsilon$ and the failure probability $\delta$.

Let us begin with the most intuitive approach based on the probability estimation, which first estimates the probability distribution $\{\braket{x|\rho|x}\}_x$ using the virtual quantum resource distillation and then samples from it. 
This approach was discussed in the context of error mitigation~\cite{Liu2025quantumerrormitigationsampling}, where the mean-square error is considered.  
We provide a detailed analysis of this approach with trace distance error---a more operationally meaningful error quantifier---applicable to the general setting beyond error mitigation.

\begin{pro}
\label{pro: sampling cost for virtual}
For a quantum state written as $\rho = c_+\sigma_+-c_-\sigma_-$ with $\sigma_\pm\in\mF$ and $c_\pm\geq 0$, the probability-estimation-based method realizes the weak distillation of state $\rho$ with trace distance accuracy $\epsilon$ with probability at least $1-\delta$ by using
\bal
N &\leq \frac{\gamma^2}{4\epsilon^2}\bigg(2^{\frac{1}{2}H_{\frac{1}{2}}(q_x)} + \sqrt{8\log\frac{2}{\delta}}\bigg)^2
\label{eq:estimation based sample bound}
\eal
samples of $\sigma_+$ and $\sigma_-$, where $q_x =\frac{c_+}{c_++c_-}p_x^+ +\frac{c_-}{c_++c_-}p_x^-$, $p_x^\pm = \braket{x|\sigma_\pm|x}$ and $\gamma = c_++c_-$. $H_\alpha(P_x)$ is the $\alpha$-R\'enyi entropy defined as 
\bal
H_\alpha(P_x) = \frac{1}{1-\alpha}\log_2\qty(\sum_x P_x^\alpha).
\eal
\end{pro}

Proof can be found in \cref{supp:estimation based} of the Supplemental Material.
\cref{pro: sampling cost for virtual} holds for any quasi-probability decomposition $\rho=c_+\sigma_+ - c_-\sigma_-$, and it can be minimized to provide the smallest sample cost, in which case $\gamma$ coincides with the robustness resource measure~\cite{Regula2018convex} with respect to the set $\mF$.
We also note that, we have $H_\frac{1}{2}(q_x)\sim n$ in the worst case, and the sampling cost grows exponentially with $n$.
On the other hand, if we somehow know that $H_\frac{1}{2}(q_x)$ is small, the sampling cost can be suppressed accordingly.

This direct method has a major drawback. 
If we consider the scenario where the target state $\rho$ is close to a free state or even a free state itself, we could simply use $\rho$ to directly obtain a sample from the output distribution. 
However, the above method still needs to estimate the probability distribution using a substantial number of states regardless of the small negativity, resulting in discontinuous behavior around $c_-=0$ as can be seen in \eqref{eq:estimation based sample bound}.

To address this discontinuous sampling cost, we propose the main framework of this work to perform weak resource distillation.
To this end, we first realize that the core problem underlying the weak resource distillation is entirely classical---we would like to sample from $\{p_x\}_x$ while only having sampling access to other distributions $\{q_x\}_x$.
This is the standard setting addressed by the method known as rejection sampling~\cite{vonNeumann1951RandomDigits} and has also recently been employed in the context of classical simulation~\cite{Dias2025sampling}.

The idea of rejection sampling is to get a sample from the available distribution $q_x$ and accept the sample with probability proportional to the fraction $\frac{p_x}{q_x}$ of \emph{acceptance ratio}. 
In our setting, this can be done by sampling from the distribution $q_x = \frac{c_+}{c_++c_-}\braket{x|\sigma_+|x} + \frac{c_-}{c_++c_-}\braket{x|\sigma_-|x}$ and accept the sample with respect to $p_x = \braket{x|\rho|x}$.
The problem caused by the difference between our setting and the standard rejection sampling is that we do not know the value of the acceptance ratio $\frac{p_x}{q_x}$ beforehand, which prevents us from deciding whether to accept the sample.
We therefore estimate the acceptance ratio using some number of states in $\mF$.
The crucial observation we employ is that we can control the accuracy of the simulation by appropriately choosing the estimation accuracy of the acceptance ratio, while one does not need to estimate the acceptance ratio accurately when $c_-$ is small.
This flexible adjustment of the accuracy can be done because we have access to the value of $c_\pm$ while not knowing the values of $p_x$ or $q_x$. 

The sampling cost to construct a sampler based on this approach can be characterized as follows. 
We provide a detailed procedure of the protocol and the proof of \cref{thm: Sampling cost for rejection} in \cref{supp:new framework} of the Supplemental Material.

\begin{thm} \label{thm: Sampling cost for rejection}
For a quantum state written as $\rho = c_+\sigma_+-c_-\sigma_-$ with $\sigma_\pm\in\mF$ and $c_\pm\geq 0$, the weak distillation of state $\rho$ with trace distance accuracy $\epsilon$ with probability at least $1-\delta$ can be accomplished by using
\bal
N &\leq 8c_-\gamma \bigg(\frac{1+\epsilon}{\epsilon}\bigg)^2\bigg(2^{\frac{1}{2}H_\frac{1}{2}{(p_x^-)}} + \delta_1^{-\frac{1}{2}}\bigg)^2 + M(\gamma, \delta_2, \epsilon)
\label{eq:sample bound rejection}
\eal
samples of $\sigma_+$ and $\sigma_-$, where $p_x^\pm = \braket{x|\sigma_\pm|x}$, $\gamma = c_++c_-$ and $(1-\delta_1)(1-\delta_2) = 1-\delta$.
Also, M is an integer that satisfies
\bal
M(\gamma, \delta_2, \epsilon)\geq \frac{\log\qty(\frac{1}{\delta_2})}{\log\qty(\frac{(1+\epsilon)\gamma}{(1+\epsilon)c_-+\epsilon\gamma})}.\\
\eal
\end{thm}

Similar to the probability distribution estimation method, the bound can be tightened if one chooses the optimal quasi-probability decomposition minimizing the negativity $c_-$.

The first term in \eqref{eq:sample bound rejection} corresponds to the sampling cost to estimate the acceptance ratio $\frac{p_x}{q_x}$, and the second term is the general overhead incurred on the rejection sampling provided the estimated acceptance ratio, whose contribution is typically negligibly small compared to the first term. (See \cref{sup:rejection sampling review} in the Supplemental Material for details.) 
Importantly, the first term is now proportional to $c_-$, which ultimately becomes zero when $\rho$ is a free state. 
Together with the second term that approaches 1 when $c_-\to 0$, the overall cost in \eqref{eq:sample bound rejection} reduces to 1 in the limit of $c_-\to 0$. 
In this case, we get that the acceptance ratio always becomes 1, corresponding to the procedure that makes a measurement on $\rho$ and accepts any measurement outcome at all times. 
This clarifies that, unlike the direct probability-estimation-based method characterized in \cref{pro: sampling cost for virtual}, our protocol is smoothly connected to the case of $c_-=0$ and efficiently takes into account the overheard due to the negativity of quasi-probability, particularly in the regime of small negativity. 

On the other hand, our bound in \eqref{eq:sample bound rejection} comes with a polynomial scaling with respect to the failure probability $\delta$ unlike the logarithmic scaling seen in \eqref{pro: sampling cost for virtual}.
In \cref{sup: alternative bound} of the Supplemental Material, we show another bound of our framework that has an exponentially better scaling with $\delta$, which ensures our framework still inherently has $\mO\qty(\log{\frac{1}{\delta}})$ dependency, while the bound loses $c_-$ proportionality.

Our framework has a practically interesting feature: if we know either of the entropies of $\{p_x^+\}_x$ or $\{p_x^-\}_x$, we can readily apply the bound in \eqref{eq:sample bound rejection}.
The entropies may be derived if either $\sigma_+$ or $\sigma_-$ is well known, or the entropy is classically efficiently computable. This can be accomplished by decomposing either of $\sigma_\pm$ as a known state or a classical simulatable state. 
On the other hand, the probability-estimation-based method characterized in \cref{pro: sampling cost for virtual} requires the knowledge of both entropies of $\{p_x^+\}_x$ and $\{p_x^-\}_x$.

Our methods only require preparing one state at a time, similar to the virtual quantum resource distillation framework. We therefore avoid the need for joint operations across multiple copies of states, which greatly reduces the experiential difficulty. The requirement of preparing many copies of a state and performing joint operations on all of them makes conventional distillation experimentally challenging, while our framework avoids this difficulty.
Additionally, there are many no-go theorems constraining distillation protocols. For example, one cannot distill a pure state with zero error from a highly mixed state, no matter how many copies we can prepare~\cite{PhysRevLett.125.060405, PRXQuantum.3.010337, regula2021fundamental}. Our framework does not suffer from such no-go theorems; we can ``distill'' any state as long as we know its quasi-probability decomposition.

\begin{figure*} 
\includegraphics[width=175mm]{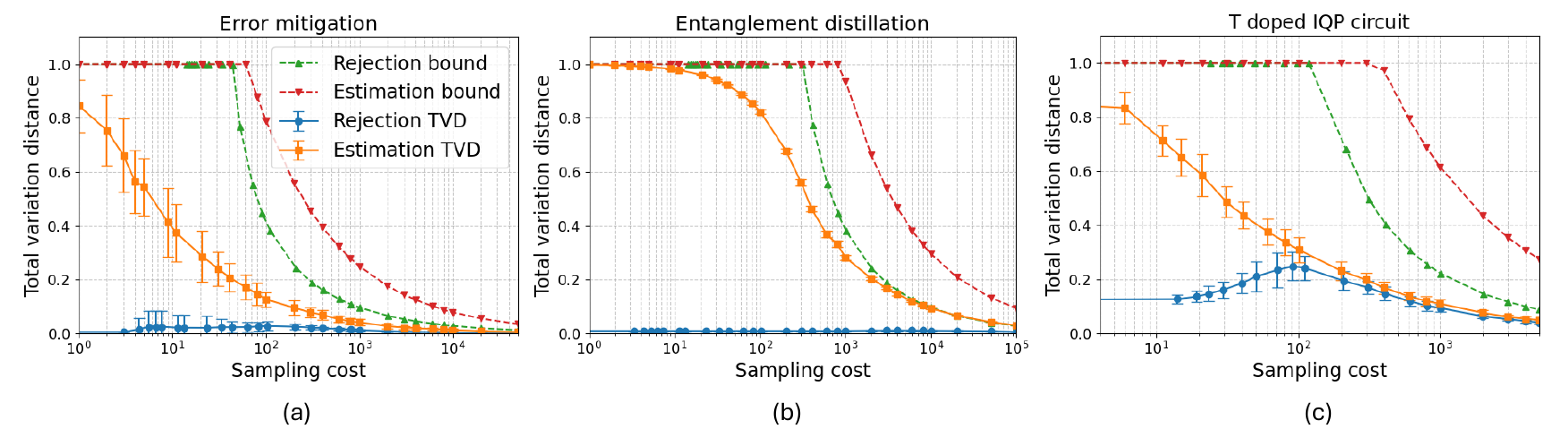}
\caption{Comparison between the total variation distance of two frameworks with respect to sampling cost. For all plots, we take the average of 20 outcomes. Blue line named ``Rejection TVD" represents the numerical simulation's outcome of the actual total variation distance we obtain with our framework. The orange line named ``Estimation TVD" represents the outcome with the probability distribution estimation method. To make a valid probability distribution, we normalized the distribution and calculated the total variation distance for the estimation method. Dotted lines represent the tightest bounds shown in the theorems. Since the total variation distance is at most 1, we take 1 as the maximum value for both bounds. We use $\delta = 0.1$. (a) The simulation of error mitigation for local depolarizing noise. For each qubit, we use the decomposition $\rho = \frac{1}{1-p}\mE(\rho) - \frac{p}{1-p}\frac{I}{2}$ where $\mE(\cdot)$ is the local depolarizing noise channel. We take $p = 0.005$, and prepared a random  $4$ qubit system. (b) The simulation of entanglement distillation. The goal is to obtain a noiseless $n$-pair of Bell states. Consider the decomposition $\Phi^{\otimes n} = \frac{1}{1-p}\rho^{n}_p- \frac{p}{1-p}\frac{I-\Phi^{\otimes n}}{2^{2n}-1}$ where $\rho^n_\alpha = (1-\alpha)\Phi^{\otimes n} + \alpha \frac{I - \Phi^{\otimes n}}{2^{2n}-1}$ is the isotropic state, we simulate the setting for 10 qubit case with $p = 0.01$. (c) The simulation of the 5 qubit $T$-doped IQP circuit, which has 5 $T$ gates in the circuit. With the dephased $T$ state defined as $\rho^T_p = (1-p)T + p \frac{I}{2}$, and dephased flipped $T$ state $\rho^{\bar{T}}_p$, we know the decomposition $T = \frac{2-p}{2(1-p)}\rho^T_p - \frac{p}{2(1-p)}\rho^{\bar{T}}_p$. We implement the $T$ gate with $T$-state injection, with the decomposition parameter set to $p = 0.1$.}
\label{fig: plots}
\end{figure*}

\vspace{0.2cm}
\textbf{\emph{Applications}.\,---}
Having the general framework in place, we now consider specific scenarios that are of practical interest and compare our method based on rejection sampling and the intuitive method based on probability estimation.

Error mitigation is of central interest for near-term quantum devices. One application of our frameworks is probabilistic error cancellation~\cite{PhysRevLett.119.180509,Endo2018practical,Takagi2021optimal}.  Let $\mE(\cdot)$ be the effective error channel of the whole circuit. If one could apply the inverse channel, then the output would just be the ideal noiseless circuit we wanted to apply. 
This cannot be done generally, as the inverse noise map is commonly not a CPTP map.
For any output state $\rho$ of the noisy circuit, when we can come up with the decomposition of the inverse channel of $\mE(\cdot)$ as $
\mE^{-1}(\rho) = c_+ \Lambda_+(\rho) - c_- \Lambda_-(\rho)$ $(  c_\pm \ge 0, c_+-c_-=1)$ where $\Lambda_\pm$ is CPTP map so that it is physically implementable, our framework allow to weak simulate $\mE^{-1}(\rho)$ by taking $\sigma_\pm = \Lambda_\pm(\rho)$. 
In \cref{fig: plots}(a), we consider mitigating local depolarizing noise and plot analytical upper bounds in \cref{pro: sampling cost for virtual} and \cref{thm: Sampling cost for rejection}, as well as actual sampling cost numerically obtained by simulating both methods. 
We also apply our framework to entanglement distillation (Fig.~\ref{fig: plots}(b)).
In the case of entanglement distillation, we aim to prepare a Bell state $\ket{\Phi} = \frac{1}{\sqrt{2}}(\ket{00} + \ket{11})$. 
Although separable states are usually considered as free states in the theory of entanglement~\cite{PhysRevA.53.2046}, we consider the setting where our experiment can produce isotropic states, $\rho_p^n = (1-p)\Phi^{\otimes n}+ p \frac{I-\Phi^{\otimes n}}{2^{2n}-1}$ and can perform LOCC.
We also consider the scenario of magic state distillation, which is of major importance for fault-tolerant quantum computation~\cite{Veitch_2014, Howard2017application}. The goal of this task is to prepare an ideal  $\ket{T} = \frac{1}{\sqrt{2}}(\ket{0} + e^{\frac{i\pi}{4}}\ket{1})$, which enables us to implement $T$ gate through a teleportation gadget. 
As in Refs.~\cite{Piveteau2021error,Suzuki2022quantum}, we consider the setting where partially dephased $T$ states are available.
We then consider sampling from an IQP circuit~\cite{Shepherd:2010kal, 10.1098/rspa.2008.0443}, where we can only prepare a dephased $T$ state instead of an ideal one. The sampling cost to approximately sample from the output distribution can be found in \cref{fig: plots}(c).

Interestingly, all scenarios in Fig. \ref{fig: plots} share common behavior. 
One is that our protocol has a lower total variation distance and a better upper bound than the probability-estimation-based method in the studied settings, particularly in regions with small sample sizes. The reason is that we cannot construct meaningful probability distributions for the probability distribution estimation method with such a small sampling size. On the other hand, our framework already allows us to access the probability distribution $\{\frac{c_+}{c_++c_-}p_x^+ + \frac{c_-}{c_++c_-}p_x^-\}_x$ without any sampling cost. For small $c_-$, that is already a good approximation of the target distribution, and therefore it provides a better outcome.

The other is that, for our framework, we observe that the total variation distance increases as the sampling number increases from 0, but it eventually decreases. By definition, the acceptance ratio changes drastically when the sample size is small. It starts from 1, but can be 0 by getting only one sample from $\sigma_-$. This sudden change sometimes makes the deviation between the ideal and empirical rates larger. As the number of samples increases, this change becomes smoother and ultimately improves the total variation distance. The extent of this bump-shaped behavior depends on the setting.  
When $c_\pm$ are large, the initial probability distribution is far from the target distribution, and obtaining a sample can improve the total variation distance, even when the acceptance ratio changes substantially. Also, when $\{\braket{x|\sigma_-|x}\}_x$ has a large entropy, the jump of the acceptance ratio is likely to happen everywhere. 

\vspace{0.2cm}
\textbf{\emph{Conclusions}.\,---} 
We introduced a framework of weak resource distillation, providing sampling access to a general class of quasi-probability methods, which previously only allowed for expectation value estimation.
Our framework inherits the wide applicability of the quasi-probability-based method, including entanglement distillation, classical simulation, error mitigation, and circuit knitting.
We show that our method based on rejection sampling has distinct features such as directly scaling with the negativity of the quasi-probability, which drastically reduces the sampling cost compared to the direct approach based on probability estimation.
We further supported the practical relevance of our framework through numerical simulation, where we demonstrated that our approach significantly outperforms the direct probability-estimation-based approach in many relevant settings, particularly reflecting the efficient construction of our method in terms of the negativity of quasi-probability.

Given the gap between the upper bound and the practical observed performance, a future alleyway is to tighten the analytical bounds. 
Also, all results presented in this work deal implicitly with finite-dimensional convex sets of available states. Nonetheless, there are many instances where quasi-probability distributions are used in infinite-dimensional systems~\cite{PhysRevA.97.062337,PhysRevA.98.052350,PRXQuantum.6.010330,xmtw-g54f}, motivating the extension of our framework to infinite-dimensional systems or systems with non-convex sets of free states~\cite{salazar2025quantumresourcetheoriesconvexity}.

 \let\oldaddcontentsline\addcontentsline
\renewcommand{\addcontentsline}[3]{}

\begin{acknowledgments}
This work is supported by JST CREST Grant Number JPMJCR23I3, JSPS KAKENHI Grant Number JP24K16975, JP25K00924, and MEXT KAKENHI Grant-in-Aid for Transformative
Research Areas A ``Extreme Universe” Grant Number JP24H00943.
 \end{acknowledgments}

\bibliography{myref_1}

\clearpage
\newgeometry{hmargin=1.2in,vmargin=0.8in}
\twocolumngrid 
\widetext  
\begin{center}
  \textbf{\large Supplemental Material}
\end{center}
\vspace{0.5cm}

\section{List of symbols}

\begin{tabular}{@{}l p{20cm}@{}}
$\Prob(X=x)$ &\quad Probability of obtaining a random variable $X$ as $x$ \\
$\Prob(X=x|Y)$ &\quad Conditional probability given $Y$\\
$\mbE[X]$ & \quad Expectation value of the random variable $X$ \\
$\mbE[X|Y]$ & \quad Conditional expectation value of $X$ for given $Y$ \\
$\mbV[X]$& \quad Variance of the random variable $X$
\\
$\mbV[X|Y]$& \quad Conditional variance value of $X$ for given $Y$\\
$\mbox{Cov}(X,Y)$&\quad Covariance between X and Y\\
$||\cdot||_1$& \quad Trace norm\\
$\left\lceil x \right\rceil$ & \quad The ceiling function, the minimum integer no less than $x$ \\
$\gamma \coloneq c_++c_-$\\
$p_x^\pm \coloneq \braket{x|\sigma_\pm|x}$\\
$p_x \coloneq c_+\braket{x|\sigma_+|x}- c_-\braket{x|\sigma_-|x} $\\
$q_x \coloneq\frac{c_+}{c_++c_-}\braket{x|\sigma_+|x}+ \frac{c_+}{c_++c_-}\braket{x|\sigma_-|x} $\\
\end{tabular}

\setcounter{section}{0}
\renewcommand{\thesection}{\arabic{section}}

\section{Probability estimation with virtual quantum resource distillation}
\label{sup: virtual distillation framework}

In this section, we introduce the general idea of virtual quantum resource distillation~\cite{Yuan2024virtual, Takagi2024virtual}. The aim of this framework is to estimate the expectation value of an observable $M$. Suppose our state of interest is $\rho'$ and that we can access another state $\rho$ which is $\beta$ close to $\rho'^{\otimes m}$, i.e. $\frac{1}{2}||\rho-\rho'^{\otimes m}||_1\le\beta$. 
Then, for an observable $M$ with $-\frac{I}{2}\le M \le \frac{I}{2}$, which can be satisfied by normalization, we know that
\bal
 \frac{1}{2}\abs{\Tr[M^{\otimes m}\rho]-\Tr[M^{\otimes m}\rho'^{\otimes m}]}&\le \frac{1}{2}\norm{\rho-\rho'^{\otimes m}}_1 \le \beta.
\eal
This fact ensures that if we can estimate the expectation value $\Tr[M^{\otimes m}\rho]$ within the error of $\epsilon$, we can estimate the target expectation value $\Tr[M^{\otimes m}\rho']$ within the error of $\beta+\epsilon$.

Now, we recap the procedure of simulating the expectation value $\Tr[M^{\otimes m}\rho]$. We denote the set of free states as $\mF$ and the set of free operations as $\mO$. Any free state that is acted by a free operation is also a free state, i.e. $\Lambda(\rho) \in \mF$ for all $\rho \in \mF, \Lambda\in \mO$.  We assume the state $\rho$, which is $\beta$ closed to $\rho'^{\otimes m}$, can be written using quasi-probability decomposition as 
\bal
\rho = c_+ \sigma_+ - c_- \sigma_- \quad(\sigma_\pm \in \mF, c_\pm \ge 0).
\label{quasi probability decomposition for rho}
\eal
We can then estimate the expectation value using this decomposition as follows.
The procedure starts with flipping a biased coin that lands heads with probability $\frac{c_+}{c_++c_-}$ and tails with probability $\frac{c_-}{c_++c_-}$. When the coin lands heads, prepare $\sigma_+$ and measure $m$ commuting observable $M\otimes I^{\otimes m-1}$,$I \otimes M\otimes I^{\otimes m-2}$, ...,$I^{\otimes m-1}...M$, multiply all values by $\gamma =c_++c_-$, and store all the measurement outcome in the classical register. When the coin lands tails, use $\sigma_-$ and measure the $m$ observables and multiply all values by $-\gamma$, store all values in the classical register. Repeat this procedure and take the sample average of the outcomes stored in the classical register. 
Then, the Hoeffding bound~\cite{Hoeffding01031963} suggests that the number of samples required to estimate $\Tr[M\rho']$ within the error of $\beta + \epsilon$ with at least success probability $1-\delta$ is upper bounded by $\mO((\gamma^2/m)\log(1/\delta)/ \epsilon^2)$.

In particular, if we take the observable $M$ as a projector onto the computational basis $\ketbra{x}{x}$, we can estimate the probability distribution $\{p_x\}_x = \{\braket{x|\rho|x}\}_x$. This can be used for the probability distribution estimation method mentioned in the main text. 

\section{Proof of Proposition~\ref{pro: sampling cost for virtual}}
\label{supp:estimation based}
In this section, we prove \cref{pro: sampling cost for virtual} of the main text. 
We repeat the proposition for the reader's convenience.

\newcounter{thmtemp}
\setcounter{thmtemp}{\value{thm}}
\setcounter{thm}{0}
\begin{pro}
For a quantum state written as $\rho = c_+\sigma_+-c_-\sigma_-$ with $\sigma_\pm\in\mF$ and $c_\pm\geq 0$, the probability-estimation-based method realizes the weak distillation of state $\rho$ with trace distance accuracy $\epsilon$ with probability at least $1-\delta$ by using
\bal
N &\leq \frac{\gamma^2}{4\epsilon^2}\bigg(2^{\frac{1}{2}H_{\frac{1}{2}}(q_x)} + \sqrt{8\log\frac{2}{\delta}}\bigg)^2
\label{eq:estimation based sample bound}
\eal
samples of $\sigma_+$ and $\sigma_-$, where $q_x =\frac{c_+}{c_++c_-}p_x^+ +\frac{c_-}{c_++c_-}p_x^-$, $p_x^\pm = \braket{x|\sigma_\pm|x}$ and $\gamma = c_++c_-$. $H_\alpha(P_x)$ is the $\alpha$-R\'enyi entropy defined as 
\bal
H_\alpha(P_x) = \frac{1}{1-\alpha}\log_2\qty(\sum_x P_x^\alpha).
\eal
\end{pro}
\setcounter{thm}{\value{thmtemp}}

\begin{proof}
This can be shown by using ideas presented in  Ref.~\cite{canonne2020short}. However, contrary to Ref.~\cite{canonne2020short}, we are interested in the setting of quasi-probability distributions, which yields both positive and negative outcomes by sampling. 
Here, we extend the analysis to this setting.

From Ref.~\cite[Sec. 1]{canonne2020short}, we know
\bal
\mbE[d_{\rm TV}(p, \tilde{p})] \le \frac{1}{2} \sum_{x} \sqrt{\mbE[(p_x-\tilde{p}_x)^2]}
\eal
where $\tilde{p}_x$ is empirical distribution that we estimate.
We get 
\bal
\mbE[(p_x-\tilde{p}_x)^2] &= \mbV[\tilde{p}_x] \\
&= \frac{\gamma^2}{N}\left(\frac{c_+}{c_++c_-}p_x^+ + \frac{c_-}{c_++c_-}p_x^- - \left(\frac{c_+}{c_++c_-}p_x^++\frac{c_-}{c_++c_-}p_x^-\right)^2\right)\\
&\le \frac{\gamma^2}{N}\left(\frac{c_+}{c_++c_-}p_x^+ + \frac{c_-}{c_++c_-}p_x^-\right) = \frac{\gamma^2}{N}q_x.
\eal
Therefore, we obtain
\bal
\mbE[d_{\rm TV}(p, \tilde{p})] \le \frac{\gamma}{2\sqrt{N}}\sum_x \sqrt{q_x}.
\eal
With $t>0$, when we take $N$ as 
\bal
N\ge \frac{\gamma^2}{4s^2}\left(\sum_x \sqrt{q_x}\right)^2,
\eal
it is ensured that $\mbE[d_{\rm TV}(p, \tilde{p})]\le s$.

Also, let $N$ be the total number of samples. Changing one sample from another cannot change the total variation distance more than $\frac{2\gamma}{N}$, which is two times as large as the one in Ref.~\cite{canonne2020short} due to the fact that we can get a negative value as a sampling outcome.
This gives 
\bal
\delta \coloneq \Prob(|d_{\rm TV}(p, \tilde{p}) - \mbE[d_{\rm TV}(p, \tilde{p})]|\ge t) &\le 2\exp\qty[-\frac{2t^2}{\sum_{n=1}^N \qty(\frac{2\gamma}{N})^2}]\\
&= 2\exp\qty(-\frac{Nt^2}{2\gamma^2})
\eal
where the dummy parameter $n$ in the denominator represents each sample. This means by taking $N$ as
\bal
N \ge \frac{2\gamma^2}{t^2}\log\frac{2}{\delta},
\eal
the inequality $|d_{\rm TV}(p, \tilde{p}) - \mbE[d_{\rm TV}(p, \tilde{p})]|\ge t$ holds at least probability $1-\delta$.
Overall, we know that, as long as we take $N$ as
\bal
N \ge \max\left\{\frac{\gamma^2}{4s^2}\left(\sum_x \sqrt{q_x}\right)^2, \frac{2\gamma^2}{t^2}\log\frac{2}{\delta}\right\}
\label{max bound},
\eal
we can estimate the probability distribution with total variation distance smaller than $t + s$ at least probability $1-\delta$. Consider minimization of the bound with respect to $r$ by setting $s = r\epsilon , t = (1-r)\epsilon$ with the target error $\epsilon$. 
This minimization is accomplished by choosing $r$'s such that both arguments inside the $\text{max\{,\}}$ are equal. By doing so, we get the upper bound in \cref{pro: sampling cost for virtual}.

\end{proof}

\section{New framework} \label{supp:new framework}
In this section, we provide a detailed explanation of our new framework based on rejection sampling~\cite{vonNeumann1951RandomDigits}.

\label{sup: rejection sampling framework}
\subsection{Rejection sampling} \label{sup:rejection sampling review}
Let us begin by reviewing the rejection sampling. Rejection sampling is a general protocol for obtaining a sample from the probability distribution $\{p_x\}_x$ using a support probability distribution $\{q_x\}_x$. We consider $q_x \neq 0$ for all $x$ such that $p_x \neq 0$.
Let us define the set of numbers $R \equiv\{R_x\}_x$ as $R_x = \frac{p_x}{K q_x}$ where $K = \sup_x \frac{p_x}{q_x}$. The first step of the protocol is to draw a sample from the support probability distribution $\{q_x\}_x$, denoted $x$.  We accept the sample with probability $R_x$. 
Otherwise, we reject the sample and get another sample from the support until the sample is accepted. 
Following this way, the accepted sample follows the target probability distribution $\{p_x\}_x$.

The last step can be seen as follows. Let $X$ be a random variable obtained by a rejection sampling framework. Also, define a random variable $Y$ as the one obtained from $\{q_x\}_x$. The probability that we get $x$ as a sample given that we know the sample gets accepted can be calculated as 
\bal
\Prob(X = x) &= \Prob(Y=x| Y \mbox{accepted})\\
&= \frac{\Prob(Y=x, Y \mbox{accepted})}{\Prob(Y \mbox{accepted})}.
\eal
Because we accept $Y$ with probability $\qty{\frac{p_x}{Kq_x}}_x$, we can evaluate the numerator as
\bal
\Prob(Y = x, Y \mbox{accepted}) &= \Prob(Y = x)\Prob( Y \mbox{accepted})\\
&=\frac{p_x}{Kq_x}q_x= \frac{p_x}{K}.
\eal
Also, we can evaluate the denominator as 
\bal
\Prob(Y \mbox{accepted}) &= \sum_x \Prob(Y = x, Y \mbox{accepted})\Prob(Y = x)\\
&= \sum_x \frac{p_x}{K}= \frac{1}{K}.
\label{eq:rejection sampling acceptance prob}
\eal
Thus, we obtain a sample of $x$ with probability $p_x$. 

Suppose that we sample $M$ times from $\{q_x\}_x$. 
Then, \eqref{eq:rejection sampling acceptance prob} tells that the probability that none of the samples is accepted is $\qty(1-\frac{1}{K})^M$. 
Therefore, a sufficient $M$ such that the sample is accepted (i.e., obtaining a sample from $\{p_x\}_x$) at least probability $\delta$ is given by  $\qty(1-\frac{1}{K})^M \le\delta$, which gives 
$M\geq \frac{\log\qty(\frac{1}{\delta})}{\log\qty(\frac{K}{K-1})}.
$ The minimum $M$ satisfy this condition is 
\bal
 M = \left\lceil\frac{\log\qty(\frac{1}{\delta})}{\log\qty(\frac{K}{K-1})}\right\rceil
  \label{eq:general sample cost rejection sampling}.
 \eal

\label{eq:general sample cost rejection sampling}

\subsection{Procedure of our framework} \label{supp:procedure}
We now provide the procedure of our framework to construct a sampler. 
The first step is to estimate the set $\{R_x\}_x$ of acceptance ratios in order to use rejection sampling.

We first prepare the lists of numbers $\{N_x^+\}_x$ and $\{N_x^-\}_x$ in classical registers. We take $N_x^\pm = 0$ for all $x$. Then we flip a biased coin that gives heads with probability $\frac{c_+}{c_++c_-}$ and tails with probability $\frac{c_-}{c_++c_-}$. When we obtain heads, we measure $\sigma_+$, get measurement outcome $x$ with probability $\braket{x|\sigma_+|x}$, and add one to $N_x^+$. Similarly, when we obtain tails, we use $\sigma_-$ and do computational basis measurement, add 1 to $N_x^-$ when we get $x$ as the measurement outcome. After repeating this process $N$ times, compute $\qty{\frac{N_x^+ - N_x^-}{N_x^+ + N_x^-}}_x$ and replace it with zero when the value is negative. 
When $N_x^+ + N_x^- = 0$, we define the rate as 1. We use this set of values as $R \equiv \{R_x\}_x$.

\subsection{The framework under an estimated acceptance ratio}
After having prepared our rejection rates $R$, we use them for rejection sampling.
In the end, we want to obtain samples from a probability distribution that is $\epsilon$ close in total variation distance to the ideal distribution with at least probability $1-\delta$.

In the following, we find the requirements to obtain the desired simulator using the set $\{R_x\}_x$ of acceptance ratios that satisfy $\frac{p_x}{Kq_x} = R_x + \epsilon_x$. Let $\{\tilde{p}_x\}_x$ be the probability distribution obtained by rejection sampling using these alternative values. 
$\{\tilde{p}_x\}_x$ can be calculated as 
\bal
\Prob(X = x)&= \frac{\Prob(Y=x, Y \mbox{accepted})}{\Prob(Y \mbox{accepted})}\\
&= \frac{R_x q_x}{\sum_y R_y q_y} = \frac{\frac{p_x}{K} - \epsilon_x q_x}{\sum_y (\frac{p_y}{K} - \epsilon_y q_y)}\\
&= \frac{p_x - K\epsilon_xq_x}{1- K\sum_y\epsilon_y q_y}.
\label{prob Xx}
\eal
Note that to consider this as a valid probability, we require $K\sum_x\epsilon_xq_x <1$. As we see later, this condition is always satisfied under the requirements we impose.  We can calculate the total variation distance between $\{p_x\}_x$ and $\{\tilde{p}_x\}_x$ as
\bal
d_{\mbox{tv}}(p, \tilde{p}) &= \frac{1}{2}\sum_x |p_x - \tilde{p}_x|\\
&= \frac{1}{2}\sum_x|p_x - \frac{p_x - K\epsilon_xq_x}{1- K\sum_y\epsilon_y q_y}|\\
&= \frac{1}{2}\sum_x|\frac{K\epsilon_xq_x+Kp_x\sum_y\epsilon_y q_y}{1- K\sum_y\epsilon_y q_y}| \le \frac{1}{2}\left(\sum_x|\frac{K\epsilon_xq_x}{1- K\sum_y\epsilon_y q_y}| + \sum_x|\frac{Kp_x\sum_y\epsilon_y q_y}{1- K\sum_y\epsilon_y q_y}|\right)\\
&\le \frac{K\sum_x|\epsilon_x| q_x}{1- K\sum_x|\epsilon_x| q_x},
\label{TVD condition for estimated rejection}
\eal
where we used the triangle inequality twice. If the last term of Eq.~\eqref{TVD condition for estimated rejection} is smaller than $\epsilon$ with probability at least $1-\delta$, then the total variation distance $d_{\mbox{tv}}$ is upper bounded by $\epsilon$ with probability at least $1-\delta$ as well, meaning we can prepare the desired simulator under this condition.
Moreover, the condition is equivalently written as 
\bal
K\sum_x|\epsilon_x|q_x \le \frac{\epsilon}{1+\epsilon}\coloneq \tilde{\epsilon}<1.
\label{condition}
\eal
This implies that Eq.~\eqref{prob Xx} is a valid probability since $K\sum_x\epsilon_x q_x < 1$ because of Eq.\eqref{condition}.

We also derive a slightly different condition for a later discussion.
In the framework we proposed, we first compute the set of values $\qty{\frac{N_x^+-N_x^-}{N_x^++N_x^-}}_x$ which could contain negative values. Note that when $N_x^++N_x^- = 0$, we take this value as 1.  
We then replace every negative value element with 0 to get the final $\{R_x\}_x$. Let $f(\cdot)$ as the function to obtain $\{R_x\}_x$ i.e. $f(\frac{N_x^+-N_x^-}{N_x^++N_x^-}) = R_x$.
We can equivalently write
\bal
\frac{p_x}{Kq_x} = \frac{N_x^+-N_x^-}{N_x^++N_x^-} + \epsilon'_x .
\eal
Because $|\epsilon'_x| \ge |\epsilon_x|$ for all $x$, we know that
\bal
d_{\rm tv}(p, \tilde{p}) \le \frac{K\sum_x|\epsilon_x| q_x}{1- K\sum_x|\epsilon_x| q_x} \le \frac{K\sum_x|\epsilon'_x| q_x}{1- K\sum_x|\epsilon'_x| q_x}.
\eal
This means that both 
\bal
K\sum_x|\epsilon_x|q_x \le \tilde{\epsilon},\quad  K\sum_x|\epsilon'_x|q_x \le \tilde{\epsilon}
\label{eq:simulator condition}
\eal
are valid conditions to make the desired simulator. We will use either of these conditions depending on the situation.

Lastly, because
\bal
K &= \sup_x \frac{c_+\braket{x|\sigma_+|x} - c_-\braket{x|\sigma_-|x}}{\frac{c_+}{c_+ + c_-}\braket{x|\sigma_+|x} + \frac{c_-}{c_+ + c_-}\braket{x|\sigma_-|x}}\\
&\le c_++c_- = \gamma
\eal
we use $K$ as $\gamma$. We might be able to use a smaller $K$, and it will give us a tighter bound on sampling cost, but for convenience we choose $K = \gamma$.

The above analysis ensure that if we take sufficiently large $N$ in \cref{supp:procedure} such that $\epsilon_x$ satisfies \eqref{eq:simulator condition}, we can construct the desired simulator. 
In the next section, we find how such $N$ can be characterized.  

After estimating the acceptance ratio, we run rejection sampling using some additional number of samples.
Suppose that we sample $M$ times and conduct rejection sampling with the determined set of values $\{R_x\}_x$. A sufficient $M$ such that at least one sample is accepted with a probability of at least $\Delta$ is given by
\bal
(1-\Prob(Y\mbox{accepted}))^M &\le \qty(1-\frac{1}{\gamma} + \sum_x \epsilon_x q_x)^M\\
&\le \qty(1-\frac{1}{\gamma} + \sum_x |\epsilon_x| q_x)^M\\
&\le \qty(1-\frac{1}{\gamma} +\tilde{\epsilon})^M.
\eal
Therefore, a sufficient $M$ such that the sample is accepted (i.e., obtaining a sample from $\{p_x\}_x$) at least probability $\delta'$ is given by
\bal
 M\geq \bigg(\log\qty(\frac{\gamma}{(1+\tilde{\epsilon})\gamma-1})\bigg)^{-1}\log\qty(\frac{1}{\delta'}) = \bigg(\log\qty(\frac{(1+\epsilon)\gamma}{(1+\epsilon)c_-+\epsilon\gamma})\bigg)^{-1}\log\qty(\frac{1}{\delta'}).
\eal
The minimum $M$ satisfying this condition is
\bal
M = \left\lceil\bigg(\log\qty(\frac{(1+\epsilon)\gamma}{(1+\epsilon)c_-+\epsilon\gamma})\bigg)^{-1}\log\qty(\frac{1}{\delta'})\right\rceil.
 \label{eq:general sample cost rejection sampling for our case}
\eal
Note that when $c_-, \epsilon \to 0$, we get that $M$ is $1$.

\section{Proof of Theorem~\ref{thm: Sampling cost for rejection}}
\label{delta theorem}
In this section, we obtain a sufficient number of samples used in the estimation stage described above. 
Together with the additional samples \eqref{eq:general sample cost rejection sampling for our case} used for running the rejection sampling, it serves as a proof of \cref{thm: Sampling cost for rejection} in the main text.
Here, we provide several forms of a sufficient number of total samples, which shows \cref{thm: Sampling cost for rejection} as its consequence.

\begin{thm}[Theorem~\ref{thm: Sampling cost for rejection} in the main text]
Let $S$ be 
\bal
S =\bigg( \sum_x \sqrt{\frac{c_+c_-p_x^+p_x^-}{c_+p_x^+ + c_-p_x^-}}\bigg)^2.
\eal
For a quantum state written as $\rho = c_+\sigma_+-c_-\sigma_-$ with $\sigma_\pm\in\mF$ and $c_\pm\geq 0$, the weak distillation of state $\rho$ with trace distance accuracy $\epsilon$ with probability at least $1-\delta$ can be accomplished by using

\bal
N &\le \min_{(1-\delta_1)(1- \delta_2)\le 1-\delta}\bigg\{8\gamma \bigg(\frac{1+\epsilon}{\epsilon}\bigg)^2 \bigg(\sqrt{S} + \sqrt{\frac{c_-}{\delta_1}}\bigg)^2 +M \bigg\}\\
&\le\min_{(1-\delta_1)(1- \delta_2)\le 1-\delta}\bigg\{8c_-\gamma \bigg(\frac{1+\epsilon}{\epsilon}\bigg)^2\bigg(2^{\frac{1}{2}H_\frac{1}{2}{(p_x^-)}} + \delta_1^{-\frac{1}{2}}\bigg)^2 + M \bigg\}\\
&\le\min_{(1-\delta_1)(1- \delta_2)\le 1-\delta}\bigg\{ 8c_+\gamma \bigg(\frac{1+\epsilon}{\epsilon}\bigg)^2\bigg( 2^{\frac{1}{2}H_\frac{1}{2}{(p_x^+)}} + \sqrt{\frac{c_-}{c_+ \delta_1}}\bigg)^2 + M\bigg\}\\
\label{bound bound}
\eal
where $\gamma = c_++c_-$, $p_x^\pm = \braket{x|\sigma_\pm|x}$ and M is an integer that satisfies
\bal
M \ge \bigg(\log\qty(\frac{(1+\epsilon)\gamma}{(1+\epsilon)c_-+\epsilon\gamma})\bigg)^{-1}\log\qty(\frac{1}{\delta_2}).
\eal
Also $H_\alpha(P_x)$ is the $\alpha$-R\'enyi entropy
\bal
H_\alpha(P_x) = \frac{1}{1-\alpha}\log_2(\sum_x P_x^\alpha).
\eal
\label{thm: Sampling cost for rejection supp}
\end{thm}

Note that the second bound in \eqref{bound bound} is used in the main text. Also, $\delta_1$ is the failure probability to prepare an $\epsilon$ close simulator, and $\delta_2$ is the failure probability of accepting at least one sample with the rejection sampling scheme. It tells us
\bal
1-\delta = (1-\delta_1)(1-\delta_2).
\eal

\begin{proof}
The second term $M$ is due to \eqref{eq:general sample cost rejection sampling for our case}, which is the number of samples generally required for rejection sampling after the estimation of the acceptance ratio $\{R_x\}_x$.
In the following, we focus on deriving the first terms, i.e., the sampling cost for estimating the acceptance ratio for rejection sampling sufficiently well.

Here, we treat $K|\epsilon'_x|q_x$ for every $x$ as a random variable in the sense that we get different $\epsilon'_x$ each time when we prepare the simulator to get $N$ samples. We denote it as $Z_x$. The first step is to check if $\{Z_x\}_x$ are negatively associated, as this later allows us to bound the variance of $\sum_x K|\epsilon'_x|q_x$ by the sum of the variance of each term. 
For an introduction to the negative association, see Ref.~\cite{Wajc2017NegativeA}.

\begin{defn}[Negative association]
A set of random variable $X_1, X_2, ..., X_n$ is negatively associated if any two disjoint index set $I, J \in [n]$ and any two functions $f, g$, both monotones increasing or both monotones decreasing, it holds
\bal
\mbE[f(X_i:i\in I)g(X_j:j\in J)] \le \mbE[f(X_i:i\in I)]\mbE[g(X_j:j\in J)]
\label{negative}
\eal
\end{defn}

In our setting, we can consider $N_x \coloneqq N_x^++N_x^-$ for all $x$ as random variables. In Ref.~\cite[Sec. 5.3]{Wajc2017NegativeA}, the example named ``Balls and Bins, and Balls and Bins" implies that $\{N_x\}_x$ are negatively associated. This fact is intuitively understood by the fact that if one $N_x$ has a large value, other random variables $N_y$ tend to have small values. According to the definition, as long as $f$ and $g$ are both monotonically increasing or both monotonically decreasing, we can use those functions for the random variables satisfying the inequality \eqref{negative}.

Here, we want to know if $\{\epsilon'_x\}_x$ satisfy Eq.~\eqref{negative}. In order to show this, we first prove the following lemma.
\begin{lem}
\label{monotones of error}
Let $\{X_1, X_2,...,X_n\}$ be random variables obtained by a Bernoulli distribution where we get $1$ with probability $p$, and get $-1$ with probability $1-p$. Also, $\theta_p(n)$ is the function that represents the absolute error of the average of random variables, i.e. 
\bal
\theta_p(n) &= \mbE\qty[\abs{\frac{1}{n}\sum_{i = 1}^n X_i - \mbE\qty[\frac{1}{n}\sum_{i = 0}^n X_i]}]\\
&= \mbE\qty[\abs{\frac{1}{n}\sum_{i = 1}^n X_i - 2p+1}]
\eal
when $n = 0$, we define $\theta_p(0) = 2(1-p)$.
Then, $\theta_p(n)$ is a monotonously decreasing function with respect to $n$, i.e.,
\bal
\theta_p(n) > \theta_p(n+1)
\eal
\end{lem}
\begin{proof}
Let $\{X_1, X_2,...,X_n, X_{n+1}\}$ be random variables obtained by a Bernoulli distribution where we get $1$ with probability $p$, and get $-1$ with probability $1-p$. Let $Y$ be a random variable obtained by choosing $j$ from $1$ to $n+1$ with equal probability and then calculating
\bal
Y = \sum_{i\neq j}^{n+1} X_i.
\eal
Even when all $\{X_1, X_2,...,X_n, X_{n+1}\}$ are determined, in the sense that we choose $j$ randomly, $Y$ is still a random variable. In this case, we know that 
\bal
\mbE[Y|X_1, X_2, ..., X_{n+1}] &=\sum_{j = 1}^{n+1} \frac{1}{n+1}\frac{1}{n}\sum_{i \neq j}^{n+1} X_i\\
&= \frac{1}{n+1} \frac{1}{n} n \sum_{i=1}^{n+1} X_i\\
&= \frac{1}{n+1} \sum_{i=1}^{n+1} X_i
\label{Y conditional}
\eal
meaning when we fix $\{X_1, X_2,...,X_n, X_{n+1}\}$, the expectation value of $Y$ is equal to the average of $\{X_1, X_2,...,X_n, X_{n+1}\}$. Because the function $|x|$ is convex with respect to $x$, for fixed $\{X_1, X_2,...,X_n, X_{n+1}\}$, \eqref{Y conditional} gives us
\bal
\abs{\frac{1}{n+1}\sum_{i = 1}^{n+1} X_i - 2p+1}&=  \abs{\mbE[Y|X_1, X_2, ..., X_{n+1}] - 2p+1}\\
&= |\mbE[Y-2p+1|X_1, X_2, ..., X_{n+1}]|\\
&<  \mbE[|Y-2p+1||X_1, X_2, ..., X_{n+1}]
\label{prob}
\eal
In the last inequality we used Jensen's inequality~\cite{Tsun2020ProbabilityStatisticsComputing}, meaning for a convex function $f$ it holds that $f(\mbE[X]) \le \mbE[f(X)]$. Note that because the function $|x|$ is not linear in $-1 \le x\le 1$, which is the range the random variables $Y-2p+1$ can take, the equality of Jensen's inequality does not hold. Finally by taking the expectation value of overall ${X_1, X_2, ...X_{n + 1}}$ on both sides of \eqref{prob} we obtain
\bal
\theta_p(n+1) = \mbE\qty[\abs{\frac{1}{n+1}\sum_{i = 1}^{n+1} X_i - 2p+1}]  &< \mbE[|Y-2p+1|]\\
&= \mbE\qty[\abs{\sum_{i \neq j}^{n+1} X_i-2p+1}]\\
&= \mbE\qty[\abs{\sum_{i=1}^{n} X_i-2p+1}]\\
&= \theta_p(n).
\eal
\end{proof}

Ref.~\cite[Sec. 5.3]{Wajc2017NegativeA} ensures that $\{N_x\}_x$ are also negatively associated. 
For all $x$, the Lemma \ref{monotones of error} says $\theta_{q_x}(n)$ as well as $K\theta_{q_x}(n)q_x$ are monotones decreasing functions with respect to $n$.  This means the random variables $\{K\theta_{q_x}(N_x)q_x\}_x$ satisfy Eq.~\eqref{negative}. This ensures that the covariance of two random variables is negative, which follows from Corollary 1 of Ref.~\cite{Wajc2017NegativeA}. This yields the following inequality
\bal
\mbV\qty[\sum_x K|\epsilon'_x|q_x] &= \mbV\qty[\sum_x K\theta_{q_x}(N_x)q_x] \\
&=\sum_x\mbV[K\theta_{q_x}(N_x)q_x] + \sum_{x, y , x \neq y}\mbox{Cov}(K\theta_{q_x}(N_x)q_x, K\theta_{q_y}(N_y)q_y) \\
&\le \sum_x\mbV[K\theta_{q_x}(N_x)q_x]=\sum_x\mbV[K|\epsilon'_x|q_x]
\label{eq:variance negative associaton}
\eal
where we use $|\epsilon'_x|=\theta_{q_x}(N_x)$ by definition. 
Next, we calculate each variance. When $N_x>0$ is fixed, the process of getting $\epsilon'_x$ is the same as calculating the error from the Bernoulli distribution. Recalling $p_x^\pm = \braket{x|\sigma_\pm|x}$, we can calculate the variance as
\bal
\mbV[\epsilon'_x|N_x] 
= \frac{4}{N_x}\frac{\frac{c_+}{c_++c_-}p_x^+}{\frac{c_+}{c_++c_-}p_x^+ + \frac{c_-}{c_++c_-}p_x^-}\frac{\frac{c_-}{c_++c_-}p_x^-}{\frac{c_+}{c_++c_-}p_x^+ + \frac{c_-}{c_++c_-}p_x^-} = \frac{4}{N_x}\frac{\frac{c_+}{c_++c_-}p_x^+\frac{c_-}{c_++c_-}p_x^-}{(\frac{c_+}{c_++c_-}p_x^+ + \frac{c_-}{c_++c_-}p_x^-)^2}.
\eal
When $N_x=0$, we define $\mbV[\epsilon'_x|N_x] = 0$.

Also, since $\mbE[\epsilon'_x|N_x] = 0$ ,
\bal
\mbV[|\epsilon'_x||N_x] &=\mbE[|\epsilon'_x|^2|N_x] -\mbE[|\epsilon'_x||N_x]^2 \\
&\le \mbE[{\epsilon'_x}^2|N_x]= \mbE[{\epsilon'_x}^2|N_x]- \mbE[\epsilon'_x|N_x]^2\\
&= \mbV[\epsilon'_x|N_x] .
\eal
Recalling that $q_x = \frac{c_+}{c_++c_-}p_x^+ + \frac{c_-}{c_++c_-}p_x^-$, the overall variance is upper bounded as 
\bal
\mbV[|\epsilon'_x|] &\le \mbV[\epsilon'_x] \\
&= \mbE_{N_x}\left[\mbV[\epsilon'_x|N_x]\right]  + \mbV_{N_x}\left[\mbE[\epsilon'_x|N_x]\right]\\
&= \sum_{N_x = 0}^N \Prob(N_x)\mbV[\epsilon'_x|N_x]\\
&= \sum_{N_x = 0}^N \begin{pmatrix}N\\N_{x}\end{pmatrix} q_x^{N_x}(1-q_x)^{N-N_x}\mbV[\epsilon'_x|N_x]\\
&= \frac{4\frac{c_+}{c_++c_-}p_x^+\frac{c_-}{c_++c_-}p_x^-}{(\frac{c_+}{c_++c_-}p_x^+ + \frac{c_-}{c_++c_-}p_x^-)^2}\sum_{n_x = 1}^N \begin{pmatrix}N\\N_{x}\end{pmatrix} q_x^{N_x}(1-q_x)^{N-N_x}\frac{1}{N_x}
\label{variance calculation}
\eal
where the first equality is due to the general relation $\mbV[X] = \mbE_Y[\mbV[X|Y]] + \mbV_Y[\mbE[X|Y]]$ with $\mbE_Y$ and $\mbV_Y$ being the expected value and variance with respect to the random variable $Y$.

To evaluate this, we first note that
\bal
\sum_{n = 0}^N \frac{1}{n+1}   \begin{pmatrix}N\\n\end{pmatrix} q^{n}(1-q)^{N-n}&=\sum_{n = 0}^N \frac{1}{n+1}\begin{pmatrix}N\\n\end{pmatrix}  q^{n}(1-q)^{N-n}\\
&=\sum_{n = 0}^N \frac{1}{n+1}\frac{N!}{(N-n)!n!} q^{n}(1-q)^{N-n}\\
&=\sum_{n = 0}^N \frac{1}{N+1}\frac{(N+1)!}{(N-n)!(n+1)!} q^{n}(1-q)^{N-n}\\
&=\frac{1}{(N+1)q}\sum_{n = 0}^{N} \begin{pmatrix}N+1\\n+1\end{pmatrix}  q^{n+1}(1-q)^{N+1-n-1}\\
&= \frac{1}{(N+1)q}(1-(1-q)^{N+1})\\
&\le \frac{1}{Nq}.
\label{1/n expectation value}
\eal
Together with the relation $\frac{1}{N_x} \le \frac{2}{N_x+1}$ for $N_x \ge 1$, Eq.~\eqref{variance calculation} can be further evaluated as
\bal
\mbV[|\epsilon'_x|] &\le \frac{4\frac{c_+}{c_++c_-}p_x^+\frac{c_-}{c_++c_-}p_x^-}{(\frac{c_+}{c_++c_-}p_x^+ + \frac{c_-}{c_++c_-}p_x^-)^2}\sum_{N_x = 1}^N \begin{pmatrix}N\\N_{x}\end{pmatrix} q_x^{N_x}(1-q_x)^{N-N_x}\frac{1}{N_x}\\
&\le \frac{4\frac{c_+}{c_++c_-}p_x^+\frac{c_-}{c_++c_-}p_x^-}{(\frac{c_+}{c_++c_-}p_x^+ + \frac{c_-}{c_++c_-}p_x^-)^2}\sum_{N_x = 1}^N \frac{2}{N_x+1}\begin{pmatrix}N\\N_{x}\end{pmatrix} q_x^{N_x}(1-q_x)^{N-N_x}\\
&\le \frac{4\frac{c_+}{c_++c_-}p_x^+\frac{c_-}{c_++c_-}p_x^-}{(\frac{c_+}{c_++c_-}p_x^+ + \frac{c_-}{c_++c_-}p_x^-)^2}\sum_{N_x = 0}^N \frac{2}{N_x+1}\begin{pmatrix}N\\N_{x}\end{pmatrix} q_x^{N_x}(1-q_x)^{N-N_x}\\
&\le \frac{4\frac{c_+}{c_++c_-}p_x^+\frac{c_-}{c_++c_-}p_x^-}{(\frac{c_+}{c_++c_-}p_x^+ + \frac{c_-}{c_++c_-}p_x^-)^2}\frac{2}{Nq_x} \\
&=\frac{8}{N} \frac{\frac{c_+}{c_++c_-}p_x^+\frac{c_-}{c_++c_-}p_x^-}{(\frac{c_+}{c_++c_-}p_x^+ + \frac{c_-}{c_++c_-}p_x^-)^3}.
\label{eq:variance upper bound}
\eal
Putting~\eqref{eq:variance negative associaton} and~\eqref{eq:variance upper bound} together, we obtain
\bal
\mbV\qty[\sum_x K|\epsilon'_x|q_x] &\le \sum_xK^2{q_x}^2\frac{8}{N}\frac{\frac{c_+}{c_++c_-}p_x^+\frac{c_-}{c_++c_-}p_x^-}{(\frac{c_+}{c_++c_-}p_x^+ + \frac{c_-}{c_++c_-}p_x^-)^3}\\
&= \frac{8(c_++c_-)^2}{N}\sum_x\frac{\frac{c_+}{c_++c_-}p_x^+}{\frac{c_+}{c_++c_-}p_x^+ + \frac{c_-}{c_++c_-}p_x^-}\frac{c_-}{c_++c_-}p_x^-\\
&\le \frac{8(c_++c_-)^2}{N}\sum_x\frac{c_-}{c_++c_-}\braket{x|\sigma_-|x}\\
&= \frac{8c_-(c_++c_-)}{N}.
\eal
On top of that, we know that
\bal
\mbE[K|\epsilon'_x|q_x] = Kq_x\mbE[|\epsilon'_x|] &=Kq_x \sqrt{\mbE[|\epsilon'_x|^2]-\mbV[|\epsilon'_x|]}\\
&= Kq_x\sqrt{\mbV[\epsilon'_x]-\mbV[|\epsilon'_x|]}\\
&\le Kq_x\sqrt{\mbV[\epsilon'_x]}\\
&\le Kq_x\sqrt{\frac{8}{N}\frac{\frac{c_+}{c_++c_-}p_x^+\frac{c_-}{c_++c_-}p_x^-}{(\frac{c_+}{c_++c_-}p_x^+ + \frac{c_-}{c_++c_-}p_x^-)^3}}\\
&\le \gamma\sqrt{\frac{8}{N}\frac{\frac{c_+}{c_++c_-}p_x^+\frac{c_-}{c_++c_-}p_x^-}{\frac{c_+}{c_++c_-}p_x^+ + \frac{c_-}{c_++c_-}p_x^-}}.
\label{eq:expectation value bound}
\eal
Using this fact and Chebyshev's inequality ~\cite{Tsun2020ProbabilityStatisticsComputing} we obtain the following bound with $t > 0$
\bal
\delta_1 \coloneqq \Prob\qty(\sum_x K|\epsilon'_x|q_x - \mbE\qty[\sum_x K|\epsilon'_x|q_x] > t) &\le
\Prob\qty(\abs{\sum_x K|\epsilon'_x|q(x) - \mbE\qty[\sum_x K|\epsilon'_x|q_x]} > t)\\&\le \frac{\mbV[\sum_x K|\epsilon'_x|q_x]}{t^2}\\
&\le \frac{8c_-\gamma}{Nt^2}.
\label{eq: condition t}
\eal
When $N \ge \frac{8c_-\gamma}{t^2\delta}
\label{1}$, then the bound implies that $|\sum_x K|\epsilon'_x|q_x - \mbE[\sum_x K|\epsilon'_x|q_x]|$ is smaller than $t$ at least probability $1-\delta_1$. 

We also bound the expected value. Together with Eq.~\eqref{eq:expectation value bound}, choosing $N$ as
\bal
N &\ge \frac{8\gamma^2}{s^2} \left(\sum_x \sqrt{\frac{\frac{c_+}{c_++c_-}p_x^+\frac{c_-}{c_++c_-}p_x^-}{\frac{c_+}{c_++c_-}p_x^+ + \frac{c_-}{c_++c_-}p_x^-}}\right)^2\\
&= \frac{8\gamma}{s^2} \left(\sum_x \sqrt{\frac{c_+c_-p_x^+p_x^-}{c_+p_x^+ + c_-p_x^-}}\right)^2
\label{eq: }
\eal
for the given $s>0$ ensures that 
\bal
s\ge \sum_x \gamma\sqrt{\frac{8}{N}\frac{\frac{c_+}{c_++c_-}p_x^+\frac{c_-}{c_++c_-}p_x^-}{\frac{c_+}{c_++c_-}p_x^+ + \frac{c_-}{c_++c_-}p_x^-}} \ge \mbE[K|\epsilon'_x|q_x].
\label{eq: s condition}
\eal

Overall, taking $t + s = \tilde{\epsilon}$ to be the bound of the total variation distance, and $N$ satisfying both \eqref{eq: condition t} and \eqref{eq: s condition}, meaning
\bal
N \ge \max\left\{ \frac{8c_-\gamma}{t^2\delta_1 }, \frac{8\gamma}{s^2}\left(\sum_x \sqrt{\frac{c_+c_-p_x^+ p_x^-}{c_+p_x^+ + c_-p_x^-}} \right)^2\right\},
\label{bound}
\eal
we can ensure that
\bal
 \Prob\left(\sum_x K|\epsilon'_x|q_x\ge\tilde\epsilon\right)\leq \delta_1.
\eal
Eq.~\eqref{bound} is a sufficient condition in order to obtain a $(\epsilon, \delta)$-simulator for our framework, as was Eq.~\eqref{condition}. By taking $s = (1-r)\tilde{\epsilon}$, $t = r\tilde{\epsilon}$, and minimizing the bound with respect to $r$, i.e. $\min_{r}\max\left\{ \frac{8c_-\gamma}{r^2\tilde{\epsilon}^2\delta_1 }, \frac{8\gamma}{(1-r)^2\tilde{\epsilon}^2}\left(\sum_x \sqrt{\frac{c_+c_-p_x^+ p_x^-}{c_+p_x^+ + c_-p_x^-}} \right)^2\right\}$. This minimization is accomplished by choosing $r$s such that both arguments inside the $\text{max\{,\}}$ are equal. 
This gives the first term in the first inequality of \eqref{bound bound}. 
The statement is concluded by noting that $(1-\delta_1)(1-\delta_2)$ is the probability that we successfully get a valid sample after running the rejection sampling, whose sampling cost $M$ is given by \eqref{eq:general sample cost rejection sampling for our case}.

To get the second and third inequalities in \eqref{bound bound}, note that
\bal
\frac{c_+c_-p_x^+ p_x^-}{c_+p_x^+ + c_-p_x^-}  = \gamma\frac{\frac{c_+}{c_++c_-}p_x^+\frac{c_-}{c_++c_-}p_x^-}{\frac{c_+}{c_++c_-}p_x^+ + \frac{c_-}{c_++c_-}p_x^-} \le c_-p_x^- \le c_+p_x^+
\eal
for every $x$. 
This gives
\bal
S = \bigg(\sum_x \sqrt{\frac{c_+c_-p_x^+ p_x^-}{c_+p_x^+ + c_-p_x^-}} \bigg)^2 \le c_-  2^{H_\frac{1}{2}(p_x^-)} \le c_+  2^{H_\frac{1}{2}(p_x^+)},
\eal
resulting in the second and third inequalities.

\end{proof}

\section{Alternative bound for our framework}
\label{sup: alternative bound}

In addition to \cref{thm: Sampling cost for rejection} in the main text, we also have an alternative upper bound of the sample complexity. 
\begin{thm}
For a quantum state written as $\rho = c_+\sigma_+-c_-\sigma_-$ with $\sigma_\pm\in\mF$ and $c_\pm\geq 0$, the weak distillation of state $\rho$ with trace distance accuracy $\epsilon$ with probability at least $1-\delta$ can be accomplished by using 
\bal
N &\le \min_{(1-\delta_1)(1- \delta_2)\le 1-\delta}\bigg\{2\gamma^2\qty(\frac{1+\epsilon}{\epsilon})^2\qty(2^{\frac{1}{2}H_{\frac{1}{2}}\qty(q_x)} + \qty(\frac{2}{v}\log\frac{1}{\delta_1})^{\frac{1}{2}})^2 +M \bigg\}
\eal
with $v = 1-e^\frac{1}{2}$ and $q_x =\frac{c_+}{c_++c_-}p_x^+ +\frac{c_-}{c_++c_-}p_x^-$. Also, M is an integer that satisfies
\bal
M \ge \bigg(\log\qty(\frac{(1+\epsilon)\gamma}{(1+\epsilon)c_-+\epsilon\gamma})\bigg)^{-1}\log\qty(\frac{1}{\delta_2}).
\eal
\label{thm: Sampling cost for rejection with log delta}
\end{thm}

The main difference between this bound and \cref{thm: Sampling cost for rejection} in the main text is that this bound has an exponentially better scaling with the failure probability bound, but loses the explicit $c_-$ dependency. The only advantage of this bound compared to the estimation-based method is the number inside the logarithm. This implies that the upper bound in the estimation method is lower than the upper bound presented in this section for most cases.
 However, this bound suggests that our framework can have a $\mO(\log\frac{1}{\delta})$ scaling, which ensures our framework is still useful for scenarios where a small $\delta$ is required. 

\begin{proof}

Contrary to the previous section, we focus on $\epsilon_x$, which is the error of the final set of numbers $\{R_x\}_x$ from the ideal set of values $\qty{\frac{p_x}{Kq_x}}_x$.
First of all, for fixed $N_x$ and $t > 0$, it holds that
\bal
\Prob(K|\epsilon_x|q_x\ge t|N_x)&\le\Prob(K|\epsilon'_x|q_x\ge t |N_x)\\
&\le 2\exp\qty(\frac{-N_xt^2}{2K^2q_x^2}).
\label{exponential bound}
\eal
In the last inequality, we used the Hoeffding inequality~\cite{Hoeffding01031963, Tsun2020ProbabilityStatisticsComputing}. By considering the possible outcomes of $K|\epsilon_x|q_x$, we obtain
\bal
K|\epsilon_x|q_x &\le \max\{2c_-\braket{x|\sigma_-|x}, c_+\braket{x|\sigma_+|x} -c_-\braket{x|\sigma_-|x}\} \\&\le c_+\braket{x|\sigma_+|x}+c_-\braket{x|\sigma_-|x} = Kq_x. 
\eal
Obviously, when $t\ge Kq_x$ , 
\bal
\Prob(K|\epsilon_x|q_x>t|N_x) = 0\quad (t>Kq_x).
\label{prob0}
\eal

Considering every possible $N_x$ (even including $N_x = 0$) within the range of $0\le t\le Kq_x$, we get
\bal
\Prob(K|\epsilon_x|q_x>t) &= \sum_{k=0}^N \Prob(N_x=k)\Prob(K|\epsilon_x|q_x>t|N_x)\\
&\le 2\sum_{k=0}^N \begin{pmatrix}N\\k\end{pmatrix} q_x^k(1-q_x)^{N-k}\exp\qty(\frac{-kt^2}{2K^2q_x^2})\\
&\le 2\sum_{k=0}^N \begin{pmatrix}N\\k\end{pmatrix} \qty[q_x \exp\qty(\frac{-t^2}{2K^2q_x^2})]^k(1-q_x)^{N-k}\\
&= 2\qty[1-q_x+q_x\exp\qty(\frac{-t^2}{2K^2q_x^2})]^N\quad (t < Kq_x)
\eal
where we use Eq.~\eqref{exponential bound} in the second line. Also, because $q_x$ is the probability of getting $x$ as the measurement outcome, we know that $\Prob(N_x = k) = \begin{pmatrix}N\\k\end{pmatrix} q_x^k(1-q_x)^{N-k}$. 
Because we consider $t \in [0, Kq_x]$, by taking $y =\frac{t^2}{K^2q_x^2}$ for $y\in [0, 1]$, we get
\bal
1-q_x+q_xe^{-\frac{y}{2}} &\le 1 -q_x(1-e^{-\frac{1}{2}})y\le \exp\left(-q_x(1-e^{-\frac{1}{2}})y\right)\\
&= \exp\left(-q_x(1-e^{-\frac{1}{2}})\frac{t^2}{K^2q_x^2}\right)\\
&= \exp\left(-\frac{vt^2}{K^2q_x}\right)
\eal
where we take $v =1-e^{-\frac{1}{2}} = 0.393... $. Then, we obtain
\bal
\Prob(K|\epsilon_x|q_x > t) \le\Prob(K|\epsilon '_x|q_x > t) \le 2\exp\left(-\frac{vNt^2}{K^2q_x}\right).
\label{eq:lemma condition 2}
\eal
We can use this bound even when $t>Kq_x$ because in that case, the probability is equal to zero see Eq.~\eqref{prob0}.

Our goal here is to combine the probability of all $x$. To do that, we use the following lemma. 
\begin{lem}\label{subgausian}
Let $X$ be a random variable that satisfies
\bal
\Prob(|X|>\epsilon)<2\exp\qty(\frac{-\epsilon^2}{2\sigma^2})
\eal
with $\sigma^2 \in \mbR$ which does not depend on $\epsilon$. Also, let $a  = \sqrt{2\sigma^2\log2}$ and consider
\bal
\mbE[|X|-a]\le 0.
\label{eq:mean negative}
\eal
Then, for $t>0$, it holds that 
\bal
\mbE[e^{t(|X|-a)}] < e^{4t^2\sigma^2}.
\eal
\end{lem}
\begin{proof}
Taking $a = \sqrt{2\sigma^2\log2}$, we get 
\bal
\Prob(|X|-a>\epsilon)&<2\exp\qty(\frac{-(\epsilon+a)^2}{2\sigma^2})\\
&=2\exp\qty(\frac{-(\epsilon^2+2a\epsilon+a^2)}{2\sigma^2})\\
&\le 2\exp\qty(\frac{-(\epsilon^2+a^2)}{2\sigma^2})\\
&=  2\exp\qty(\frac{-\epsilon^2}{2\sigma^2})\exp\qty(\frac{-2\sigma^2\log2}{2\sigma^2})\\
&=\exp\qty(\frac{-\epsilon^2}{2\sigma^2}).
\eal
Let $k\ge0$ be an integer. 
When $k$ is an odd number, following the proof of Ref.~\cite[Lemma 1.4]{rigollet2023highdimensionalstatistics}, we obtain
\bal
\mbE((|X|-a)^k) &= \int_0^\infty\Prob((|X|-a)^k>t) dt - \int_0^\infty\Prob((|X|-a)^k<-t) dt\\
&\leq \int_0^\infty \Prob(|X|-a>t^{\frac{1}{k}})dt\\
&\leq \int_0^\infty \exp\qty(-\frac{t^\frac{2}{k}}{2\sigma^2}) dt\\
&= \frac{1}{2}(2\sigma^2)^\frac{k}{2}k\varGamma\qty(\frac{k}{2})\\
&\le (2\sigma^2)^\frac{k}{2}k\varGamma\qty(\frac{k}{2})
\eal
where $\varGamma(\cdot)$ denotes the Gamma function. Because $\Prob(|X|<t) \le 1$ for $t\ge0$ and $\Prob(|X|<t) =0$ for $t<0$, when $k$ is even number, we know
\bal
\mbE((|X|-a)^k) &= \int_0^\infty\Prob((|X|-a)^k>t) dt\\
&= \int_0^\infty\left( \Prob(|X|-a>t^{\frac{1}{k}}) +  \Prob(|X|-a<-t^{\frac{1}{k}})\right)dt\\
&\leq \int_0^\infty \exp\qty(-\frac{t^\frac{2}{k}}{2\sigma^2}) dt + \int_0^{a^k}1dt\\
&= \frac{1}{2}(2\sigma^2)^\frac{k}{2}k\varGamma\qty(\frac{k}{2}) + a^k\\
&= \frac{1}{2}(2\sigma^2)^\frac{k}{2}k\varGamma\qty(\frac{k}{2}) + (2\sigma^2)^\frac{k}{2}(\log2)^\frac{k}{2}\\
&\le (2\sigma^2)^\frac{k}{2}k\varGamma\qty(\frac{k}{2}).
\eal
In the last inequality, we used the fact that $\frac{1}{2}k\varGamma(\frac{k}{2})> (\log2)^\frac{k}{2}$. The reason is that $(\log2)^\frac{k}{2}$ is a monotonically decreasing function and for $k\ge2$, $\frac{1}{2}k\varGamma(\frac{k}{2})$ is monotonically increasing. When $k = 1, 2$, we can get $\frac{1}{2}k\varGamma(\frac{k}{2})> (\log2)^\frac{k}{2}$ by numerical calculation. Therefore, for all $k$, $\frac{1}{2}k\varGamma(\frac{k}{2})> (\log2)^\frac{k}{2}$ holds.

Overall, analogously to Lemma 1.5 in Ref ~\cite{rigollet2023highdimensionalstatistics}, we get
\bal
\mbE[e^{s(|X|-a)}] &= 1 + \mbE[|X|-a]s + \sum_{k = 2}^\infty\frac{s^k\mbE[(|X|-a)^k]}{k!}\\
&\le 1 + \mbE[|X|-a]s + \sum_{k = 2}^\infty \left(\frac{(2\sigma^2s^2)^\frac{k}{2}k\varGamma(\frac{k}{2})}{k!} \right)\\
&\le 1  + \sum_{k = 2}^\infty \frac{(2\sigma^2s^2)^\frac{k}{2}k\varGamma(\frac{k}{2})}{k!}\\
&\le \exp({4\sigma^2s^2}) 
\eal
 where in the third line, we used the assumption \eqref{eq:mean negative}.
\end{proof}

Now we prove the main theorem. We define
\bal
a(x) = \sqrt{\frac{2K^2q_x}{N}} \ge \sqrt{\frac{K^2q_x}{vN}\log2}.
\eal
Note that $2\ge \frac{log2}{v} = 1.76...$. Then we get 
\bal
\Prob\qty(\big|K|\epsilon_x|q_x-a(x)\big|>t)
&\le\Prob(K|\epsilon_x|q_x-a(x)>t)=\Prob(K|\epsilon_x|q_x>a(x)+t)\\&<2\exp\qty(\frac{-vN(t+a(x))^2}{K^2q_x})\\
&= 2\exp\qty(\frac{-vN(t^2+2t a(x)+a(x)^2)}{K^2q_x})\\
&\le 2\exp\qty(\frac{-vN(t^2+a(x)^2)}{K^2q_x})\\
&\le 2\exp\qty(\frac{-vN(t^2+\frac{K^2q_x}{vN}\log2)}{K^2q(x)})\\
&= 2\exp\qty(\frac{-vNt^2}{K^2q_x})\exp\qty(\frac{-vN\frac{K^2q_x}{vN}\log2}{K^2q_x})\\
&= \exp\qty(\frac{-t^2}{2\frac{K^2q_x}{2vN}}).
\eal

Also, since for $a,b>0$, we know that $
\frac{8ab}{a+b} \le 2(a+b)$. This implies that
\bal
\mbE[K|\epsilon_x|q_x]&\le\mbE[K|\epsilon'_x|q_x]\\
&\le  \gamma\sqrt{\frac{8}{N}\frac{\frac{c_+}{c_++c_-}p_x^+\frac{c_-}{c_++c_-}p_x^-}{\frac{c_+}{c_++c_-}p_x^+ + \frac{c_-}{c_++c_-}p_x^-}}\\
&\le  \gamma\sqrt{\frac{2}{N}\left(\frac{c_+}{c_++c_-}p_x^+ + \frac{c_-}{c_++c_-}p_x^-\right)}\\
&= a(x)
\eal
and therefore
\bal
\mbE[K|\epsilon_x|q_x-a(x)]\le 0.
\label{eq: lemma condition}
\eal
Because of Eq.~\eqref{eq:lemma condition 2} and Eq.~\eqref{eq: lemma condition} we can now apply Lemma~\ref{subgausian} for our case as
\bal
\mbE\qty[\exp({s\sum_x(K|\epsilon_x|q_x-a)})]&\le \mbE\qty[\exp({s\sum_x(K|\epsilon'_x|q_x-a)})]\\
&\le \prod_x \mbE[\exp({s(K|\epsilon'_x|q_x-a)})]\\
&\le \prod_x \exp{\left(4\frac{K^2q_x}{2vN}s^2\right)}\\
&= \exp\left(\frac{2s^2K^2}{vN}\sum_xq_x \right)\\
&= \exp\left(\frac{2s^2K^2}{vN} \right)
\eal
where the second inequality comes from the fact $\{K|\epsilon'_x|q_x\}_x$ are negatively associated, which we showed after the proof of Lemma \ref{monotones of error}. By using the Chernoff bound~\cite{10.1214/aoms/1177729330} which suggests $\Prob(X>t)<\mbE[e^{sX}]e^{-st}$ for  any $s > 0$, we get
\bal
\delta_1 = \Prob\qty(\sum_x(K|\epsilon_x|q_x-a(x)) \ge t) &\le \Prob\qty(\sum_x(K|\epsilon'_x|q_x-a(x)) \ge t)\\
&<\mbE\qty[\exp(s\sum_x(K|\epsilon'_x|q_x-a) \ge t)]e^{-st}\\
&\le \exp\left(\frac{s^2K^2}{vN}-st \right)\\
&= \exp\left(\frac{-vt^2N}{4K^2}\right) \quad(s = \frac{vt N}{2K^2}).
\eal
This implies the fact that when $N$ satisfies 
\bal
N \ge \frac{4\gamma^2}{vt^2}\log\frac{1}{\delta_1}
\label{bound1},
\eal
$\sum_x(K|\epsilon_x|q_x-a(x))$ gets smaller than $t$ with probability at least $1-\delta_1$. In addition to that, for 
\bal
s \ge \sum_x a(x) = \sum_x \gamma\sqrt{\frac{2}{N}q_x}
\eal
or equivalently
\bal
N \ge \frac{2\gamma^2}{s^2}\left(\sum_x \sqrt{q_x}\right)^2,
\label{bound2}
\eal
it holds that $|\mbE[\sum_x K |\epsilon_x|q_x]-\sum_x a(x)| \le s$. With the same discussion as in the previous section, we upper bound the sampling cost by using Eq.\eqref{bound1} and Eq.\eqref{bound2} and taking $\tilde{\epsilon} = t + s$.
 Similar to the proof of \cref{pro: sampling cost for virtual}, taking $t = r\tilde{\epsilon}, s = (1-r)\tilde{\epsilon}$, we get $\min_r\max\{\frac{4\gamma^2}{vr^2\tilde{\epsilon}^2}\log\frac{1}{\delta_1}, \frac{2\gamma^2}{(1-r)^2\tilde{\epsilon}}\left(\sum_x \sqrt{q_x}\right)^2\}$. This minimization is accomplished by choosing $r$'s such that both arguments inside the $\text{max\{,\}}$ are equal. Finally, by considering the failure probability $\delta_2$ of getting at least one accepted sample from the rejection sampling scheme, we get the statement.
\end{proof}

\end{document}